\documentclass{article}

\usepackage[preprint]{neurips_2023}

\usepackage[utf8]{inputenc} 
\usepackage[T1]{fontenc}    
\usepackage{hyperref}       
\usepackage{url}            
\usepackage{booktabs}       
\usepackage{amsfonts}       
\usepackage{nicefrac}       
\usepackage{microtype}      
\usepackage{xcolor}         

\usepackage{bm} 
\usepackage{amsmath,amssymb} 
\usepackage{amsthm}  
\usepackage{tabularx}
\usepackage{graphicx}

\usepackage[italicdiff]{physics}

\newtheorem{theorem}{Theorem}[section]

\newtheorem{prop}[theorem]{Proposition}

\DeclareMathOperator*{\argmin}{arg\,min}
\DeclareMathOperator*{\argmax}{arg\,max}

\title{What is the most optimal diffusion?}

\author{%
  Vasili Baranau
  \texttt{vasili.baranov@gmail.com} \\
}

\begin{document}

\maketitle

\begin{abstract}
    What is the fastest possible ``diffusion''? 
    A trivial answer would be ``a process that converts a Dirac delta-function 
    into a uniform distribution infinitely fast''. 
    Below, we consider a more reasonable formulation: 
    a process that maximizes differential entropy of 
    a probability density function (pdf) $f(\vec{x}, t)$ at every time $t$, under certain restrictions. 
    Specifically, we focus on a case when the rate of the 
    Kullback--Leibler divergence $D_{\text{KL}}$ is fixed.
    If $\Delta(\vec{x}, t, \dd{t}) = \frac{\partial f}{ \partial t} \dd{t}$ 
    is the pdf change at a time step $\dd{t}$, 
    we maximize the differential entropy $H[f + \Delta]$ 
    under the restriction $D_{\text{KL}}(f + \Delta || f) = A^2 \dd{t}^2$, $A = \text{const} > 0$.
    It leads to the following equation: 
    $\frac{\partial f}{ \partial t} = - \kappa f (\ln{f} - \int f \ln{f} \dd{\vec{x}})$, 
    with $\kappa = \frac{A}{\sqrt{ \int f \ln^2{f} \dd{\vec{x}} - \left( \int f \ln{f} \dd{\vec{x}} \right)^2 } }$.
    Notably, this is a non-local equation, 
    so the process is different from the It\^{o} diffusion 
    and a corresponding Fokker--Planck equation.
    We show that the normal and exponential distributions are 
    solutions to this equation, on $(-\infty; \infty)$ and $[0; \infty)$, respectively, 
    both with $\text{variance} \sim e^{2 A t}$, \textit{i.e.} diffusion is highly anomalous. 
    We numerically demonstrate for sigmoid-like functions on a segment 
    that the entropy change rate $\frac{d H}{d t}$ produced by such an optimal ``diffusion'' 
    is, as expected, higher than produced by the ``classical'' diffusion.
    %
\end{abstract}


\section{\label{sec:Introduction} Introduction}

Diffusion processes and random walks are ubiquitous in nature and 
technology; many areas of science study them: 
physics, chemistry, econometrics, 
and others \cite{kampen_stochastic_2007, enders_applied_2014}. 
Recently, diffusion found prominent application in machine learning as a 
basis of diffusion-based models for image and video 
generation \cite{sohl_dickstein_diffusion_2015, ho_denoising_diffusion_2020, rombach_latent_diffusion_2022}.

In this paper, we study the question: 
what is the most optimal, the fastest possible ``diffusion''?
Our motivation is two-fold: firstly, diffusion models in 
machine learning modify diffusion rate with time to achieve 
desired ``speed'' of feature generation. 
Hence, using the most ``optimal'' diffusion can be 
a natural way to improve these architectures.
Secondly, we see this question as the extension 
of the old ``brachistochrone'' problem: what is the shape of the 
curve to optimize the time for moving along this 
curve in a gravitational field from point $A$ to point $B$.
This problem extended the notion of the optimum from functions 
(taking a derivative) to functionals (taking a functional, or variational, derivative), 
and the question about the fastest possible ``diffusion'' can 
extend the concept of an optimum to operators. 

A trivial answer to this question is 
``a process that converts a Dirac delta-function 
into a uniform distribution infinitely fast'', 
but it is not a very fruitful solution.
A more general formulation can be this: we need to find an operator 
(belonging to a certain class of operators) that 
(a) on average minimizes the time to increase the differential 
entropy of a probability density function (pdf) by a certain amount (b) for a 
given class of initial pdfs, (c) under 
restrictions on how ``intensive'' diffusion can be 
(e.g. how much ``energy'' is poured into the system).
If one considers only spatially local operators, which can be 
formulated as $O = \sum_i a_i(t) \frac{\partial^i}{\partial x^i}$, 
one can convert the problem to finding an optimal set of $a_i(t)$, 
i.e. calculating $\frac{\partial L}{\partial a_i}$ 
for a certain optimization function $L$. 

In this paper, we focus on a different special case of the general problem: 
we do not impose a constraint of spatial locality, 
but require that the differential entropy of a 
pdf increases in an optimum way on every time ``step'' of function evolution.
To make the solution more interesting than 
``from any pdf to the uniform pdf infinitely quickly'', 
we still need to limit the rate of ``spreading'' the function, 
or, alternatively, limit the ``energy'' that is being poured into the system.

Thus, if $f(\vec{x}, t)$  is a pdf of 
a continuous probability distribution, 
and if $\Delta(\vec{x}, t, \dd{t})$ is the pdf change at a time step $\dd{t}$, 
we want to maximize the differential entropy $H[f + \Delta]$.
As an ``energy'' restriction we choose restricting 
the Kullback--Leibler divergence $D_{\text{KL}}(f + \Delta || f)$, 
$D_{\text{KL}}(f + \Delta || f) = A^2 \dd{t}^2$, $A = \text{const}$ 
(\textit{cf}. Proposition \ref{prop:KullbackLeiblerThroughDerivative} below for details).
Since the total probability shall be unity at all $t$, we need to maintain additional restrictions, 
$\int f \dd{\vec{x}} \equiv 1$ and $\int \Delta \dd{\vec{x}} \equiv 0$. 
It allows to formulate the problem as a simple variational calculus problem, with the Lagrangian 
\begin{equation*}
    L[\Delta] = 
      H[f + \Delta] 
      - \lambda \left( D_{\text{KL}}(f + \Delta || f) - A \dd{t} \right) 
      -\mu \int \Delta \dd{\vec{x}} = \text{max}, 
\end{equation*}
where $\lambda$ and $\mu$ are Lagrangian multipliers.

\subsection{\label{subsec:RelatedWork} Related Work}

Anomalous diffusion and L{\'e}vy flights are extensively studied in 
physics and chemistry literature \cite{chechkin_introduction_levy_flights_2008}.

After the rise in popularity of diffusion models 
in machine learning \cite{sohl_dickstein_diffusion_2015, 
ho_denoising_diffusion_2020, rombach_latent_diffusion_2022}, 
several authors investigated certain versions of optimal diffusion and random walks.

\paragraph{Optimal It\^{o} diffusion.}

Ref. \cite{jafarizadeh_optimal_diffusion_2018} suggests how to optimize the functions 
$\mu(x)$ (expectation, or drift) and $\sigma(x)$ 
(where $D(x) = \sigma^2 / 2$ is the variance, or the diffusion coefficient) in the 
stochastic differential equation $\dd{X(t)} = \mu(X(t)) \dd{t} + \sigma(X(t)) \dd{W(t)}$, 
where $W(t)$ is the Wiener process, 
with the objective to minimize the time to reach 
the desired stationary distribution $\pi(x)$ from a given distribution, 
under the additional constraint that the average diffusion coefficient is predefined, 
$\int \pi(x) \frac{1}{2}\sigma^2(x) \dd{x} = \frac{1}{2}\hat{\sigma}^2$.
This paper focuses on traditional stochastic differential equations in $\mathbb{R}^1$. 
Hence, so that the optimal solution in terms of a pdf $f(x, t)$ 
still conforms to a Fokker--Planck equation 
$\frac{\partial}{\partial t} f(x, t) = 
  -\frac{\partial}{\partial x} \left[ \mu(x) f(x, t) \right] 
  +\frac{\partial^2}{\partial x^2} \left[ D(x) f(x, t) \right]$.
They approach the problem by optimizing 
the second largest eigenvalue (the largest is zero) of a certain operator, 
which minimizes diffusion time. 
They provide semi-analytical solutions in a one-dimensional case.

Ref. \cite{biswal_spectral_gap_optimization_2020} investigates a problem very 
similar to \cite{jafarizadeh_optimal_diffusion_2018}, 
they similarly optimize the second largest eigenvalue (the largest is zero) of a certain operator, 
to minimize the diffusion time of a Fokker--Planck equation to a target distribution.
The do not limit the problem to a one-dimensional case, 
and provide a numerical method instead of a semi-analytical solution.

Different to the papers above, in the present manuscript, 
we explicitly focus on extensions to the Fokker--Planck equations 
that can potentially be spatially non-local.

\paragraph{Optimal control theory and It\^{o} diffusion.}

Ref. \cite{chertovskih_optimal_control_of_diffusion_2024} 
analyzes the problem of optimally controlling an It\^{o} diffusion process. 
That is, they analyze the problem of indirectly controlling drift and diffusion coefficients 
through control parameters so that a given cost functional (which depends on system trajectories) 
is on average (over system trajectories) minimized.
Again, the authors consider It\^{o} diffusion, 
and extend the diffusion problem with the optimal control setting.

Ref. \cite{berner_optimal_control_vs_diffusion_2023} investigates the connection between 
diffusion probabilistic models and stochastic optimal control theory.

\paragraph{Poisson Flow Generative Models.}

Refs. \cite{xu_poisson_flow_2022, xu_poisson_flow_plus_plus_2023} introduce 
diffusion-like deep learning models for image generation based on the Poisson equation or 
Maxwell equations for electrodynamics instead of the diffusion equation. 
A potential extension of the present manuscript is to train a diffusion-like model 
based on the equation derived below, taking it instead of the Poisson equation, 
since our equation shall by construction perform ``diffusion'' optimally.

\subsection{\label{subsec:Contributions} Contributions}

\begin{itemize}
    \item The present manuscript posits, to our knowledge, a novel problem of finding a law of optimal ``diffusion'', 
      which does not have to conform to the It\^{o} diffusion and the Fokker--Planck equation
    \item We provide an explicit equation for a particular formulation of such an ``diffusion'' process
    \item We provide several special solutions in $\mathbb{R}^1$: 
      We show that the normal and exponential distributions are 
      solutions to our equation, on $(-\infty; \infty)$ and $[0; \infty)$, respectively, 
      both with $\text{variance} \sim e^{2 A t}$, \textit{i.e.} diffusion is highly anomalous.
      The truncated normal distribution is a solution on a segment, 
      with a non-trivial parameterized equation.
\end{itemize}

The paper is structured as follows: in Section \ref{sec:GeneralFormulation}, 
we present a general formulation of the problem. 
In Section \ref{sec:LocallyOptimumDiffusion}, 
we present a formulation for diffusion, 
which is ``locally optimum in time'', and derive an explicit equation.
In Section \ref{sec:SpecialSolutions}, we demonstrate several special solutions to the equation.
Section \ref{sec:SimulationResults} provides simulation results 
and compares our equation to the classical diffusion on a segment. 
Finally, Section \ref{sec:Summary} provides conclusions and outlook.

\section{\label{sec:GeneralFormulation} General formulation}

We denote as an optimal ``diffusion'' operator the operator that on average optimizes diffusion time.
We use the following notation and assumptions.

\paragraph{Operator.}

Let $O$ be the operator that is being optimized. 
It acts on spatial coordinates $\vec{x} \in \mathbb{R}^N$ of 
probability density functions $f(\vec{x}, t): \mathbb{R}^{N} \times \mathbb{R} \to \mathbb{R}$, 
$O: (\mathbb{R}^{N} \to \mathbb{R}) \to (\mathbb{R}^{N} \to \mathbb{R})$.

We posit the dynamics as 
\begin{equation}
    \frac{\partial f}{\partial t} = O f.
\end{equation}
This notation implies that we can formally express the pdf at time $t$ as 
\begin{equation*}
    f(\vec{x}, t) = e^{O t} f_0,
\end{equation*}
where $f_0$ belong to the distribution of initial probability density functions 
$f_0 = f(\vec{x}, t_0)$. One restriction on $f_0$ is $\int f_0 \dd{\vec{x}} = 1$.

\paragraph{Optimization criteria.}

We denote with $T$ the target functional (acting on spatial variables) 
that describes how ``spread'' is the probability density function. 
For example, it could be the differential entropy or variance. 
It shall be maximum at the uniform distribution.

We seek to optimize diffusion time, 
\textit{i.e.} the time for $T[f]$ to reach a certain value or change by a given amount, 
\begin{equation*}
    T[f(\vec{x}, t_2)] = T_{\text{final}}. 
\end{equation*}
As an example, one can posit $T_{\text{final}} = 0.1 T[f_0] + 0.9 T[f_{\text{uniform}}]$.

\paragraph{Restricting the ``energy flow''.}

To arrive at reasonable solutions, 
\textit{i.e.} not an infinitely fast diffusion, 
we need to restrict the ``energy flow'' to the system, 
implying $\frac{d}{\dd{t}} E[f] = \text{const}$ for some functional $E$, 
also acting on spatial variables only. 
A slightly more general formulation is to 
specify the functional for the ``energy flow'' itself, 
\begin{equation*}
    F[f, \frac{\partial f}{\partial t}] = \text{const}
\end{equation*}
for some functional $F$.
It could be based on the Kullback--Leibler divergence or the earth mover's distance.

We need to find an operator $O$ that is optimum on average, 
over a distribution of initial probability density functions 
$f_0$.

\paragraph{Well-behaving operators.}

$O f$ shall always produce a sufficiently differentiable function. 
Thus, 
\begin{equation*}
    e^{O t} f_0 \in C^1,
\end{equation*}
$\forall t \geq 0$ and $\forall$ $f_0$ from the distribution of initial pdfs.

The requirement $\int f \dd{\vec{x}} \equiv 0 ~ \forall t$ implies that 
$\int \frac{\partial f}{\partial t} \dd{\vec{x}} \equiv 0 ~ \forall t \geq 0$ and 
$\int O f \dd{\vec{x}} \equiv 0 ~ \forall t \geq 0$. Thus,
\begin{equation*}
    \int O e^{O t} f_0 \dd{\vec{x}} \equiv 0,
\end{equation*}
$\forall t \geq 0$ and $\forall$ $f_0$ from the distribution of initial pdfs.

We also imply that 
\begin{equation*}
    \frac{d T[e^{O t} f_0]}{d t} \geq 0,
\end{equation*}
$\forall t \geq 0$ and $\forall$ $f_0$ from the distribution of initial pdfs.

Additionally, we would be interested in operators 
belonging to a certain class $O \in C$. 
One example is a class of operators that are spatially local.
$C$ is always a subset of 
\begin{equation}
\begin{aligned}
    C_0 = \Bigl\{
        O: & \int O e^{O t} f_0 \dd{\vec{x}} \equiv 0, \\
            & e^{O t} f_0 \in C^1, \\
            & \frac{d T[e^{O t} f_0]}{d t} \geq 0, \\
            & F[f, O e^{O t} f_0] = \text{const}, \\
            & ~ \forall t \geq 0, \forall \text{~valid~} f_0 \Bigl\}.
\end{aligned}
\end{equation}

With this notation and assumptions, we arrive at the following formulation
\begin{equation}
\begin{aligned}
    O =& 
    \argmin_{O \in C} \left[ 
        \mathbb{E}_{f_0} \left[ 
            t_2: T[e^{O t_2} f_0] = T_{\text{final}}
        \right]
    \right], C \subset C_0.
  \label{eq:GeneralFormulationOperator}
\end{aligned}
\end{equation}
The result will depend on the distribution of $f_0$.

\section{\label{sec:LocallyOptimumDiffusion} ``Diffusion'' that is locally optimal in time}

In the remainder of the paper, we focus on a specific non-trivial example: 
operators $O$ are supposed to optimize the target functional $T$ at every time $t$, \textit{i.e.}
``diffusion'' is locally optimum (with respect to time). 
What makes this case interesting is that the result is not local with respect to $\vec{x}$, 
\textit{i.e.}, the resulting ``diffusion'' does not conform to the Fokker--Planck equation.

Similar to the classical brachistochrone problem, 
such a local (in time) solution may not be the globally optimum operator. 
Operators that lead to slower ``diffusion'' at the beginning 
but ``prepare'' the function for a very fast ``diffusion'' at a 
second stage can still win and on average be optimum operators.
We do not focus on this general case in this paper.

Formally, we require that operators $O$ belong to the class $C = C_0 \cup C_{\text{opt}}$ with
\begin{equation*}
    C_{\text{opt}} = \left\{O: O f = g_0, ~\forall f,
        \text{~where~} g_0 = \argmax_g \frac{T[f + g \dd{t}] - T[f]}{\dd{t}}, \dd{t} \to 0 \right\}.
\end{equation*}
It makes Eq. (\ref{eq:GeneralFormulationOperator}) independent 
of the distribution of $f_0$ and of $T_{\text{final}}$, and we can simply write the problem as 
\begin{equation}
\begin{aligned}
    O : O f =& g_0, ~ \forall f,\\
    g_0 =& \argmax_{g \in C_g} \frac{T[f + g \dd{t}] - T[f]}{\dd{t}}, \dd{t} \to 0, \\
    C_g =& \left\{ g: \int g \dd{\vec{x}} = 0, 
            g \in C^1, 
            F[f, g] = \text{const} \right\}.
  \label{eq:LocallyOptimumOperator}
\end{aligned}
\end{equation}

We formulate the problem in a slightly more convenient way now. 
At an arbitrary $t$, $f$, and $\dd{t} \to 0$, we define the change in $f$ as 
\begin{equation}
    \Delta(\vec{x}, t, \dd{t}) := \frac{\partial f}{ \partial t} \dd{t}, 
\end{equation}
so 
\begin{equation}
    O f := \frac{\Delta(\vec{x}, t, \dd{t})}{\dd{t}}.
\end{equation}
We aim to optimize the the Lagrangian with respect to $\Delta(\vec{x}, t, \dd{t})$. 
The optimization problem is constrained, 
since we need to maintain $F[f, \frac{\partial f}{\partial t}] = \text{const} = A$, 
as well as $\int \Delta \dd{\vec{x}} \equiv 0$ to make sure that $\int f \dd{\vec{x}} \equiv 1$.
Thus, we can write the Lagrangian as
\begin{equation}
    L[\Delta] = 
        T[f + \Delta] 
        - \lambda \left(F[f, \frac{\Delta}{\dd{t}}] - A \right) 
        - \mu \left( \int \Delta \dd{\vec{x}} - 0 \right),
  \label{eq:LocallyOptimumLagrangianGeneric}
\end{equation}
where $\lambda$ and $\mu$ are Lagrangian multipliers.

For the remainder of the paper, we posit $T$ to be the differential entropy $H$, and 
use $D_{\text{KL}}(f + \Delta || f)$ to restrict the energy flow. 
Note that $D_{\text{KL}}(f + \Delta || f) \sim \dd{t}^2$ 
(\textit{cf}. Proposition \ref{prop:KullbackLeiblerThroughDerivative} below), 
so we specify the restriction as
\begin{equation}
    D_{\text{KL}}(f + \Delta || f) 
    = D_{\text{KL}}(f + \frac{\partial f}{\partial t} \dd{t} || f) 
    = A^2 \dd{t}^2,
  \label{eq:RestrictionOnKLDivergence}
\end{equation}
implying that 
$F = 
\left( 
    \frac{D_{\text{KL}}(f + \frac{\partial f}{\partial t} \dd{t} || f)}{\dd{t}^2} 
\right)^{1/2}$.

Thus, we arrive at the final Lagrangian
\begin{equation}
    L[\Delta] = 
      H[f + \Delta] 
      - \lambda \left( D_{\text{KL}}(f + \Delta || f) - A^2 \dd{t}^2 \right) 
      -\mu \int \Delta \dd{\vec{x}} = \text{max}. 
  \label{eq:Lagrangian}
\end{equation}

One can use other combinations of functionals for measuring the ``spread'' and ``energy flow'', 
like variance and the earth mover's distance, but we do not study them in this manuscript.

\subsection{\label{subsec:OptimizingLagrangian} Main result: Optimizing the Lagrangian}

We optimize the Lagrangian with respect to $\Delta$ using variational calculus. 
The solution shall conform to
\begin{equation*}
    \frac{\delta L}{\delta \Delta} = 0,
    \frac{\partial L}{\partial \lambda} = 0,
    \frac{\partial L}{\partial \mu} = 0.
\end{equation*}

For the remainder of the paper, proofs of propositions are in the appendix.

\begin{prop}\label{prop:DerivativeOfLagrangian}
    Let $f$ and $\Delta$ be non-negative 
    sufficiently differentiable and integrable functions,
    additionally conforming to $\int f \dd{\vec{x}} = 1$ and $\int (f + \Delta) \dd{\vec{x}} = 1$. 
    Then maximizing the Lagrangian (\ref{eq:Lagrangian}) leads to 
    \begin{equation}
        \frac{\partial f}{\partial t} = - \kappa f \left[ \ln f - \int f \ln f \dd{\vec{x}} \right]
        \label{eq:OptimalDiffusionWithKappa}
    \end{equation}
    for some constant $\kappa$.
\end{prop}

Let us explicitly show how $\kappa$ maps to $A$ from Eq. (\ref{eq:Lagrangian}).

\begin{prop}\label{prop:KullbackLeiblerThroughDerivative}
    Let $f$ and $\Delta$ be non-negative sufficiently differentiable and integrable functions 
    conforming to $\int f \dd{\vec{x}} = 1$ and $\int (f + \Delta) \dd{\vec{x}} = 1$. 
    Additionally, let $\Delta(\vec{x}) \ll f(\vec{x}) ~ \forall \vec{x}$. Then,
    \begin{equation}
    \begin{aligned}
        D_{\text{KL}}(f + \Delta || f) 
        =& \int \frac{\Delta^2}{f} \dd{\vec{x}} \\
        D_{\text{KL}}(f + \frac{\partial f}{\partial t} \dd{t} || f)
        =& (\dd{t})^2 \int \frac{1}{f} \left( \frac{\partial f}{\partial t} \right)^2 \dd{\vec{x}}.
        \label{eq:KlChangeRate}
    \end{aligned}
    \end{equation}
\end{prop}

\begin{prop}\label{prop:ValueOfKappa}
    Let $f$ be a non-negative sufficiently differentiable and integrable function 
    conforming to $\int f \dd{\vec{x}} = 1$, as well as to Eq. (\ref{eq:OptimalDiffusionWithKappa}). Then,
    \begin{equation}
        \kappa = \frac{
            \sqrt{\frac{D_{\text{KL}}(f + \frac{\partial f}{\partial t} \dd{t} || f)}{\dd{t}^2}}
        } {
            \sqrt{ \int f \ln^2 f \dd{\vec{x}} - \left( \int f \ln f \dd{\vec{x}} \right)^2 }
        }
        \label{eq:KappaFromA}
    \end{equation}
\end{prop}

Generally, $\kappa$ and 
$\sqrt{\frac{D_{\text{KL}}(f + \frac{\partial f}{\partial t} \dd{t} || f)}{\dd{t}^2}}$ 
can be functionals of $f$ and functions of $t$, 
the proof of Eq. (\ref{eq:OptimalDiffusionWithKappa}) remain intact. 
It means that one can enforce time dependence in 
$\sqrt{\frac{D_{\text{KL}}(f + \frac{\partial f}{\partial t} \dd{t} || f)}{\dd{t}^2}}$, but 
the form of Eq. (\ref{eq:OptimalDiffusionWithKappa}) remains the same.

\section{\label{sec:SpecialSolutions} Special solutions}

It is possible to derive several special solutions to 
Eqs. (\ref{eq:OptimalDiffusionWithKappa}) and (\ref{eq:KappaFromA}).

We will find it easier to interpret some special solutions when the form of $\kappa = \kappa[f]$ 
enforces the constant growth in the differential entropy,
\begin{equation}
    \frac{d H}{d t} = \text{const}.
    \label{eq:DerivativeEntropyByTimeConstant}
\end{equation}

\begin{prop}\label{prop:DerivativeEntropyByTime}
    Let $f$ be a non-negative sufficiently differentiable and integrable function 
    conforming to $\int f \dd{\vec{x}} = 1$. 
    Then,
    \begin{equation}
        \frac{d H}{d t} = - \int \ln{f} \frac{\partial f}{\partial t} \dd{\vec{x}}.
        \label{eq:DerivativeEntropyByTime}
    \end{equation}
\end{prop}

The relation for $\kappa[f]$ can then be expressed as follows.
\begin{prop}\label{prop:ValueOfKappaConstantEntropyRate}
    Let $f$ be a non-negative sufficiently differentiable and integrable function 
    conforming to $\int f \dd{\vec{x}} = 1$, as well as to Eq. (\ref{eq:OptimalDiffusionWithKappa}). Then,
    \begin{equation}
        \kappa = \frac{\frac{d H}{d t}}{ \int f \ln^2 f \dd{\vec{x}} - \left( \int f \ln f \dd{\vec{x}} \right)^2 }.
        \label{eq:KappaFromEntropyRate}
    \end{equation}
\end{prop}

The following fact will also be useful
\begin{prop}\label{prop:EntropyRateVsKullbackLeiblerRate}
    Let $f$ be a non-negative sufficiently differentiable and integrable function 
    conforming to $\int f \dd{\vec{x}} = 1$, as well as to Eq. (\ref{eq:OptimalDiffusionWithKappa}). Then,
    \begin{equation}
        \frac{D_{\text{KL}}(f + \frac{\partial f}{\partial t} \dd{t} || f)}{\dd{t}^2} 
        = \kappa \frac{d H}{d t}.
        \label{eq:EntropyRateVsKullbackLeiblerRate}
    \end{equation}
\end{prop}

This equation additionally demonstrates that entropy can only increase, 
if governed by Eq. (\ref{eq:OptimalDiffusionWithKappa}).

\subsection{\label{subsec:SpecialSolutionNormal} Normal distribution in 1D}

\begin{prop}\label{prop:NormalDistributionIsSolution}
    Let $f$ be a pdf of a general normal distribution, $f = f_N(x | 0, \sigma(t))$. 
    Then, it satisfies Eq. (\ref{eq:OptimalDiffusionWithKappa}) 
    on the entire real line if $\forall t$
    \begin{equation}
        \kappa = \frac{1}{\sigma} \frac{d \sigma}{d t} = \frac{d \ln{\sigma} }{d t}.
        \label{eq:KappaForNormalDistribution}
    \end{equation}
\end{prop}

Eq. (\ref{eq:KappaForNormalDistribution}) is valid for any restriction on the speed of 
``diffusion'', be it (\ref{eq:DerivativeEntropyByTimeConstant}) or 
(\ref{eq:RestrictionOnKLDivergence}).

\begin{prop}\label{prop:NormalDistributionEntropyIncrease}
    Let $f$ be a pdf of a general normal distribution, $f = f_N(x | 0, \sigma(t))$. 
    Let $\sigma(t)$ change in a way that $f$ conforms to Eq. (\ref{eq:OptimalDiffusionWithKappa}) with 
    $\kappa = \kappa_0 = \text{const}$. Then, 
    \begin{equation}
    \begin{aligned}
        \sqrt{\frac{D_{\text{KL}}(f + \frac{\partial f}{\partial t} \dd{t} || f)}{\dd{t}^2}} 
            = \frac{d H}{d t} = A = \kappa_0 = \text{const}.
    \end{aligned}
    \end{equation}
\end{prop}

For any normal distribution, since $H[f_N] = \frac{1}{2} + \frac{1}{2} \ln{(2 \pi \sigma^2)}$, 
it holds that $\text{Var}[f_N] = \sigma^2 \sim e^{2 H[f_N]}$, 
independent of Eq. (\ref{eq:OptimalDiffusionWithKappa}). 
Thus, under conditions of Proposition \ref{prop:NormalDistributionEntropyIncrease},
\begin{equation}
    \text{Var}[f_N] = \sigma_0^2 e^{2 A t},
\end{equation}
so ``diffusion'' generated by Eq. (\ref{eq:OptimalDiffusionWithKappa}) 
under conditions of Proposition \ref{prop:NormalDistributionEntropyIncrease} is highly anomalous.

\subsection{\label{subsec:SpecialSolutionExponential} Exponential distribution}

\begin{prop}\label{prop:ExponentialDistributionIsSolution}
    Let $f$ be a pdf of the exponential distribution, $f = f_E(x | \lambda(t)) = \lambda e^{- \lambda x}$, $x \geq 0$. 
    Then, it satisfies Eq. (\ref{eq:OptimalDiffusionWithKappa}) 
    for $x \in [0, \infty)$ if $\forall t$
    \begin{equation}
        \kappa = - \frac{1}{\lambda} \frac{d \lambda}{d t} = - \frac{d \ln{\lambda} }{d t}.
        \label{eq:KappaForExponentialDistribution}
    \end{equation}
\end{prop}

\begin{prop}\label{prop:ExponentialDistributionEntropyIncrease}
    Let $f$ be a pdf of the exponential distribution, $f = f_E(x | \lambda(t))$. 
    Let $\lambda(t)$ change in a way that $f$ conforms to Eq. (\ref{eq:OptimalDiffusionWithKappa}) with 
    $\kappa = \kappa_0 = \text{const}$. Then, 
    \begin{equation}
    \begin{aligned}
        \sqrt{\frac{D_{\text{KL}}(f + \frac{\partial f}{\partial t} \dd{t} || f)}{\dd{t}^2}} 
            = \frac{d H}{d t} = A = \kappa_0 = \text{const}.
    \end{aligned}
    \end{equation}
\end{prop}

For the exponential distribution, since $H[f_E] = 1 - \ln{\lambda}$ and $\text{Var}[f_E] = \frac{1}{\lambda^2}$, 
it holds that $\text{Var}[f_E] \sim e^{2 H[f_E]}$, independent of Eq. (\ref{eq:OptimalDiffusionWithKappa}). 
Thus, under conditions of Proposition \ref{prop:ExponentialDistributionEntropyIncrease},
\begin{equation}
    \text{Var}[f_E] = \sigma_0^2 e^{2 A t},
\end{equation}
so ``diffusion'' generated by Eq. (\ref{eq:OptimalDiffusionWithKappa}) 
under conditions of Proposition \ref{prop:ExponentialDistributionEntropyIncrease} is highly anomalous.

\subsection{\label{subsec:TruncatedNormalDistribution} Truncated normal distribution}

It is also possible to show through a symbolic math package like \textit{sympy} 
that a symmetrical truncated normal distribution will be a solution to 
Eq. (\ref{eq:OptimalDiffusionWithKappa}) on a segment. 
If one formulates the pdf as $\ln{f} = a x^2 - c$, 
the solution can be found as an implicit function for $a = a(c)$ and $t = t(c)$.

\section{\label{sec:SimulationResults} Simulation results on a segment}

We will investigate the following toy example. Given a truncated sigmoid-like function on a segment,
\begin{equation}
\begin{aligned}
    f(x | c) =& \frac{1}{Z} \left( \frac{1}{1 + e^{-\frac{x}{c}}} + 0.05 \right), \\
    Z =& \int_{-1}^{1} \left( \frac{1}{1 + e^{-\frac{x}{c}}} + 0.05 \right) \dd{x},
    \label{eq:SigmoidLikeFunction}
\end{aligned}
\end{equation}
we will compare entropy rates from the classical diffusion and the optimal ``diffusion''. 
Specifically, we will
\begin{itemize}
    \item calculate $\frac{\partial f}{\partial t}$ according to the diffusion equation 
        $\frac{\partial f}{\partial t} = \frac{\partial^2 f}{\partial x^2}$, 
    \item calculate $\sqrt{\frac{D_{\text{KL}}(f + \frac{\partial f}{\partial t} \dd{t} || f)}{\dd{t}^2}}$ 
        for such a classical diffusion dynamics,
    \item normalize $\kappa$ from Eq. (\ref{eq:OptimalDiffusionWithKappa}) so that 
        $\sqrt{\frac{D_{\text{KL}}(f + \frac{\partial f}{\partial t} \dd{t} || f)}{\dd{t}^2}}$ 
        equals to the value from the classical diffusion (to make the two examples of dynamics properly comparable),
    \item calculate $\frac{\partial f}{\partial t}$ according to the ``optimal'' diffusion, 
        Eq. (\ref{eq:OptimalDiffusionWithKappa}).
    \item calculate $\frac{d H}{d t}$ for both cases according to Eq. (\ref{eq:DerivativeEntropyByTime}).
\end{itemize}

Figure \ref{fig:EntropyRatesVsC} depicts $\frac{d H}{d t} (c)$ for both types of dynamics, 
classical diffusion and ``optimal'' diffusion according to Eq. (\ref{eq:OptimalDiffusionWithKappa}).
It demonstrates that optimal ``diffusion'' ensures, as expected, higher $\frac{d H}{d t}$.
The more abrupt the change in the test function, the higher the difference in entropy rates.
The panels provide guidance how abrupt (high-frequency) the change in the function shall be to 
lead to noticeable differences in the dynamics.

\begin{figure*}[ht]
    \centering
    \def\svgwidth{0.95\linewidth}
    \includegraphics{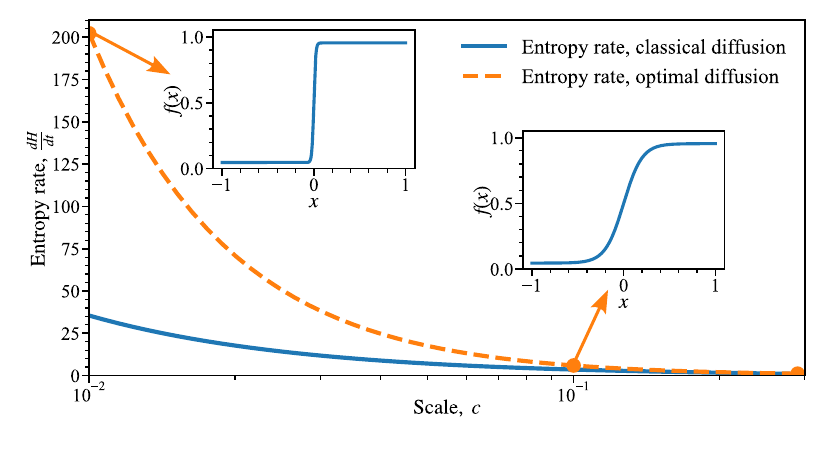}
    \caption{
        Entropy change rate $\frac{d H}{d t}$ from Eq. (\ref{eq:DerivativeEntropyByTime}) compared 
        for the classical diffusion dynamics 
        $\frac{\partial f}{\partial t} = \frac{\partial^2 f}{\partial x^2}$ and for 
        the optimum ``diffusion'', Eq. (\ref{eq:OptimalDiffusionWithKappa}). 
        Here, $f$ is the sigmoid-like function $f(x|c)$, Eq. (\ref{eq:SigmoidLikeFunction}), 
        evaluated at different scale parameters $c$. 
        The free constant $\kappa$ in Eq. (\ref{eq:OptimalDiffusionWithKappa}) is scaled 
        to ensure equal change rate of the KL-divergence, Eq. (\ref{eq:KlChangeRate}).
        The plot demonstrates that optimal ``diffusion'' ensures, as expected, higher $\frac{d H}{d t}$.
        The more abrupt the change in the test function, the higher the difference in entropy rates.
    }
    \label{fig:EntropyRatesVsC}
\end{figure*}

\section{\label{sec:Summary} Conclusions and outlook}

We provided a general formulation of the problem of optimal ``diffusion'' and 
investigated one class of solutions, ``diffusion'' locally optimum in time.
We derived an explicit differential equation (integro-differential equation) for this case, 
and provided several special solutions, 
as well as demonstrated with a numerical example that this 
equation indeed leads to a much faster ``diffusion'' than the classical diffusion equation.
These results are only initial steps in 
exploring the problem of finding the most optimal ``diffusion'', 
and there are several avenues for further research.

For Eq. (\ref{eq:OptimalDiffusionWithKappa}) specifically, 
it would be interesting to understand if there is a type of 
L{\'e}vy flights, Langevin dynamics, or an It\^{o} process 
that would lead to such a macroscopic dynamics.

For a more general optimality criterion, 
finding ``diffusion'' that is optimum non-locally in time 
(similar to the non-trivial brachistochrone curve),
it would be interesting to derive a mathematical apparatus 
to find an optimum among operators.
As a step in this direction, 
one can investigate spatially local operators in the form 
$O = \sum_i a_i(t) \frac{\partial^i}{\partial x^i}$. 
Then, the problem reduces to finding an optimal set of $a_i(t)$, 
which is a more tractable task, but still provides only a special solution.

Finally, as a practical application, it would be promising 
to train a diffusion model for image generation using 
Eq. (\ref{eq:OptimalDiffusionWithKappa}), 
similar to how Refs. \cite{xu_poisson_flow_2022, xu_poisson_flow_plus_plus_2023} 
use the Poisson equation instead of classical diffusion.


\bibliographystyle{abbrv}
\bibliography{OptimalDiffusion}

\newpage

\appendix

\section{\label{sec:OptimizingLagrangianProofs} Proofs from Section \ref{subsec:OptimizingLagrangian}}

For variational derivatives, we use the standard fact that 
\begin{equation*}
    \frac{\delta}{\delta f} \int g(f) \dd{\vec{x}} = \frac{\partial g}{\partial f},
\end{equation*}
where the partial derivative is taken as if $f$ were a simple variable.
It follows from the definition of a variational derivative: 
given a functional $F[f]$, $\frac{\delta F}{\delta f}$ is such that 
\begin{equation*}
    \delta F = F[f + \delta f] - F[f] = \int \frac{\delta F}{\delta f} \delta f.
\end{equation*}

\subsection{Preliminaries}

\begin{prop}\label{prop:DerivativeOfEntropy}
    \begin{equation}
        \frac{\delta H[f + \Delta]}{\delta \Delta} = - \left[ \ln{(f + \Delta)} + 1 \right].
    \end{equation}
\end{prop}

\begin{proof}
    \begin{equation*}
    \begin{aligned}
        \frac{\delta H[f + \Delta]}{\delta \Delta}
        &= - \frac{\delta}{\delta \Delta} \int (f + \Delta) \ln{(f + \Delta)} \dd{\vec{x}} \\
        &= - \left[ \ln{(f + \Delta)} + 1 \right].
    \end{aligned}
    \end{equation*}
\end{proof}

\begin{prop}\label{prop:DerivativeOfKullbackLeibler}
    \begin{equation}
        \frac{\delta}{\delta \Delta} D_{\text{KL}}(f + \Delta || f) 
        = \ln{(f + \Delta)} + 1 - \ln f.
        \label{eq:DerivativeOfKullbackLeibler}
    \end{equation}

\end{prop}

\begin{proof}
    \begin{equation*}
    \begin{aligned}
    \frac{\delta}{\delta \Delta} D_{\text{KL}}(f + \Delta || f) 
    &= \frac{\delta}{\delta \Delta} \int (f + \Delta) \ln{ \left( \frac{f + \Delta}{f} \right) } \dd{\vec{x}} \\
    &= \frac{\delta}{\delta \Delta} 
        \left[ \int (f + \Delta) \ln{(f + \Delta)} \dd{\vec{x}} - \int (f + \Delta) \ln(f) \dd{\vec{x}} \right] \\
    &= \ln{(f + \Delta)} + 1 - \ln f.
    \end{aligned}
    \end{equation*}
\end{proof}

\subsection{Proof of Proposition \ref{prop:DerivativeOfLagrangian}}

\begin{proof}
    \begin{equation*}
        \frac{\delta L}{\delta \Delta} = 
        - \left[ \ln{(f + \Delta)} + 1 \right] - \lambda \left[ \ln{(f + \Delta)} + 1 - \ln f \right] -\mu.
    \end{equation*}

    $\frac{\delta L}{\delta \Delta} = 0$ implies the following. We will take the minus sign of the entire equation and 
    will use the first-order approximation for $\ln{(f + \Delta)} = \ln f + \frac{\Delta}{f} + O(\Delta^2)$:

    \begin{equation*}
    \begin{aligned}
        \ln f + \frac{\Delta}{f} + 1 + \lambda \left[ \ln f + \frac{\Delta}{f} + 1 - \ln f \right] + \mu &= 0, \\ 
        \ln f + (1 + \lambda) \frac{\Delta}{f} + 1 + \lambda + \mu &= 0. 
    \end{aligned}
    \end{equation*}

    Multiplying by $f$ and integrating, utilizing $\int f \dd{\vec{x}} = 1$ and $\int \Delta \dd{\vec{x}} = 0$, gives
    \begin{equation*}
        \int f \ln f \dd{\vec{x}} + 1 + \lambda + \mu = 0.  
    \end{equation*}

    Hence, $\Delta = - f \frac{1}{1 + \lambda} \left[ \ln f - \int f \ln f \dd{\vec{x}} \right]$.

    Since $\Delta = \frac{\partial f}{\partial t} \dd{t}$, it follows that $\frac{1}{1 + \lambda} = \kappa \dd{t}$ for 
    some constant $\kappa$, and hence we arrive at Eq. (\ref{eq:OptimalDiffusionWithKappa}).

\end{proof}

\subsection{Proof of Proposition \ref{prop:KullbackLeiblerThroughDerivative}}

One way to conduct the computation is the following:

\begin{proof}
    \begin{equation*}
    \begin{aligned}
        D_{\text{KL}}(f + \Delta || f) 
        &= \int \left[ \frac{\delta}{\delta \Delta} D_{\text{KL}}(f + \Delta || f) \right] \Delta \dd{\vec{x}}  \\
        &= [\text{Eq.} ~ (\ref{eq:DerivativeOfKullbackLeibler})] \\
        &= \int (\ln{(f + \Delta)} + 1 - \ln f) \Delta \dd{\vec{x}} \\
        &= [\Delta \ll f ~ \forall \vec{x}, ~ \text{ignoring } O(\Delta^3) ~ \text{terms}] \\
        &= \int (\ln{f} + \frac{\Delta}{f} + 1 - \ln f) \Delta \dd{\vec{x}} \\
        &= \int \frac{\Delta^2}{f} \dd{\vec{x}}.
    \end{aligned}
    \end{equation*}
\end{proof}

Essentially, we have calculated 
$D_{\text{KL}}(f + \Delta || f) = D_{\text{KL}}(f || f) 
+ \int \frac{\delta D_{\text{KL}}(f_2 || f)}{\delta f_2} \mid_{f_2=f} \Delta \dd{\vec{x}} 
+ \int \frac{\delta^2 D_{\text{KL}}(f_2 || f)}{\delta f_2^2} \mid_{f_2=f} \Delta^2 \dd{\vec{x}}$.
It always holds that $D_{\text{KL}}(f || f) = 0$ and 
$\frac{\delta D_{\text{KL}}(f_2 || f)}{\delta f_2} \equiv 0$. 
It is only expected, since $D_{\text{KL}}(f_2 || f)$ 
shall be non-negative and zero if $f_2 = f$.

One can achieve the same result by a naive expansion of the 
expression under the integral in $D_{\text{KL}}(f + \Delta || f)$ by $\Delta$ up to the second order.

\subsection{Proof of Proposition \ref{prop:ValueOfKappa}}

\begin{proof}
    \begin{equation*}
    \begin{aligned}
        D_{\text{KL}}(f + \frac{\partial f}{\partial t} \dd{t} || f) 
        &= [\text{Eq.} ~ (\ref{eq:KlChangeRate})] = 
        (\dd{t})^2 \int \frac{1}{f} \left( \frac{\partial f}{\partial t} \right)^2 \dd{\vec{x}} \\
        &= [\text{Eq.} ~ (\ref{eq:OptimalDiffusionWithKappa})] \\
        &= (\dd{t})^2 \int \frac{1}{f} \kappa^2 f^2 \left( \ln f - \int f \ln f \dd{\vec{x}} \right)^2 \dd{\vec{x}} \\
        &= (\dd{t})^2 \kappa^2 \int f \left\{ \ln^2 f - 2 \ln f (\int f \ln f \dd{\vec{x}}) + (\int f \ln f \dd{\vec{x}})^2 \right\} \dd{\vec{x}} \\
        &= (\dd{t})^2 \kappa^2 \int \left\{ f \ln^2 f - 2 f \ln f (\int f \ln f \dd{\vec{x}}) + f (\int f \ln f \dd{\vec{x}})^2 \right\} \dd{\vec{x}} \\
        &= (\dd{t})^2 \kappa^2 \left\{ \int f \ln^2 f \dd{\vec{x}} - 2 (\int f \ln f \dd{\vec{x}})^2 + (\int f \ln f \dd{\vec{x}})^2 \right\} \\
        &= (\dd{t})^2 \kappa^2 \left\{ \int f \ln^2 f \dd{\vec{x}} - (\int f \ln f \dd{\vec{x}})^2 \right\} \\
        &= [\text{Eq.} ~  (\ref{eq:Lagrangian})] \\
        &= A^2 (\dd{t})^2.
    \end{aligned}
    \end{equation*}
\end{proof}

\section{\label{sec:SpecialSolutionsProofs} Proofs from Section \ref{sec:SpecialSolutions}}

\subsection{Proof of Proposition \ref{prop:DerivativeEntropyByTime}}

\begin{proof}
    \begin{equation*}
    \begin{aligned}
        \delta H =& \int \frac{\delta H}{\delta f} \delta f \dd{\vec{x}} 
        = \left[ \delta f = \frac{\partial f}{\partial t} \dd{t} \right] 
        = \int \frac{\delta H}{\delta f} \frac{\partial f}{\partial t} \dd{t} \dd{\vec{x}}, \\
        \frac{d H}{d t} =& \int \frac{\delta H}{\delta f} \frac{\partial f}{\partial t} \dd{\vec{x}}, \\
        \frac{\delta H}{\delta f} =& -\ln{f} - 1, \\
        \frac{d H}{d t} =& \left[\int \frac{\partial f}{\partial t} \dd{\vec{x}} = 0 \right] 
        = - \int \ln{f} \frac{\partial f}{\partial t} \dd{\vec{x}}.
    \end{aligned}
    \end{equation*}
\end{proof}

\subsection{Proof of Proposition \ref{prop:ValueOfKappaConstantEntropyRate}}

\begin{proof}
    \begin{equation*}
    \begin{aligned}
        \frac{d H}{d t} &= [\text{Eqs.} ~ (\ref{eq:DerivativeEntropyByTime}) \text{~and~} 
            (\ref{eq:OptimalDiffusionWithKappa})] 
            = \int \kappa f \ln{f} \left[ \ln f - \int f \ln f \dd{\vec{x}} \right] \dd{\vec{x}} \\
        &= \kappa \left[ \int f \ln^2{f} \dd{\vec{x}} - \left(\int f \ln f \dd{\vec{x}} \right)^2 \right].
    \end{aligned}
    \end{equation*}
\end{proof}

\subsection{Proof of Proposition \ref{prop:EntropyRateVsKullbackLeiblerRate}}

\begin{proof}
    \begin{equation*}
    \begin{aligned}
        \frac{1}{\int f \ln^2 f \dd{\vec{x}} - \left( \int f \ln f \dd{\vec{x}} \right)^2 } =& 
            [\text{Eq.~} (\ref{eq:KappaFromA})] 
            = \frac{\kappa^2}{D_{\text{KL}}(f + \frac{\partial f}{\partial t} \dd{t} || f) / \dd{t}^2}, \\
        =& [\text{Eq.~} (\ref{eq:KappaFromEntropyRate})] = \frac{\kappa}{d H / d t}, \\
        \frac{D_{\text{KL}}(f + \frac{\partial f}{\partial t} \dd{t} || f)}{\dd{t}^2} 
            =& \kappa \frac{d H}{d t}.
    \end{aligned}
    \end{equation*}
\end{proof}

\subsection{Proof of Proposition \ref{prop:NormalDistributionIsSolution}}

\begin{proof}
    It is convenient to reformulate Eq. (\ref{eq:OptimalDiffusionWithKappa}) as
    \begin{equation}
    \begin{aligned}
        \frac{\partial \ln{f}}{\partial t} = - \kappa \left[ \ln f + H[f] \right].
        \label{eq:OptimalDiffusionLogF}
    \end{aligned}
    \end{equation}
    
    Checking the left-hand side:
    \begin{equation*}
    \begin{aligned}
        \ln{f_N} =& -\frac{x^2}{2 \sigma^2} - \frac{1}{2} \ln{2 \pi \sigma^2}, \\
        \text{LHS} =& \frac{\partial \ln{f_H}}{\partial t} \\
        =& \frac{x^2}{\sigma^3} \frac{d \sigma}{d t} - \frac{1}{\sigma} \frac{d \sigma}{d t} \\
        =& \frac{1}{\sigma} \frac{d \sigma}{d t} \left[ \frac{x^2}{\sigma^2} - 1 \right].
    \end{aligned}
    \end{equation*}
    
    Checking the right-hand side, given that $H[f_N] = \frac{1}{2} + \frac{1}{2} \ln{(2 \pi \sigma^2)}$:
    \begin{equation*}
    \begin{aligned}
        \text{RHS} =& - \kappa \left[ \ln f_N + H[f_N] \right] \\
         =& - \kappa \left[ -\frac{x^2}{2 \sigma^2} - \frac{1}{2} \ln{(2 \pi \sigma^2)} + \frac{1}{2} + \frac{1}{2} \ln{(2 \pi \sigma^2)} \right] \\
         =& \frac{\kappa}{2} \left[ \frac{x^2}{\sigma^2} - 1 \right]
    \end{aligned}
    \end{equation*}
    
    Thus, normal distribution can conform to 
    Eq. (\ref{eq:OptimalDiffusionWithKappa}) if $\forall t$ 
    \begin{equation*}
    \begin{aligned}
        \kappa = \frac{1}{\sigma} \frac{d \sigma}{d t} = \frac{d \ln{\sigma} }{d t}.
    \end{aligned}
    \end{equation*}
    
\end{proof}

\subsection{Proof of Proposition \ref{prop:NormalDistributionEntropyIncrease}}

\begin{proof}
    For a normal distribution, 
    \begin{equation*}
    \begin{aligned}
        H[f_N] =& \frac{1}{2} + \frac{1}{2} \ln{(2 \pi \sigma^2)} = \text{const} + \ln{\sigma}, \\
        \frac{d H}{d t} =& \frac{d \ln{\sigma}}{d t}
    \end{aligned}
    \end{equation*}
    Hence, $\kappa = \kappa_0 = \text{const}$ implies, through Eq. (\ref{eq:KappaForNormalDistribution}), 
    that $\frac{d H}{d t} = \kappa_0 = \text{const}$.
    It immediately follows from Eq. (\ref{eq:EntropyRateVsKullbackLeiblerRate}) that 
    $\sqrt{\frac{D_{\text{KL}}(f + \frac{\partial f}{\partial t} \dd{t} || f)}{\dd{t}^2}} = \kappa_0$.

\end{proof}

\subsection{Proof of Proposition \ref{prop:ExponentialDistributionIsSolution}}

\begin{proof}
    Checking the left-hand side of Eq. (\ref{eq:OptimalDiffusionLogF}):
    \begin{equation*}
    \begin{aligned}
        \ln{f_E} =& \ln{\lambda} - \lambda x, \\
        \text{LHS} =& \frac{1}{\lambda} \frac{d \lambda}{d t} - \frac{d \lambda}{d t} x \\
        =& \frac{1}{\lambda} \frac{d \lambda}{d t} (1 - \lambda x).
    \end{aligned}
    \end{equation*}

    Checking the right-hand side, given that $H[f_E] = 1 - \ln{\lambda}$:
    \begin{equation*}
    \begin{aligned}
        \text{RHS} =& - \kappa \left[ \ln{\lambda} - \lambda x + 1 - \ln{\lambda} \right] \\
         =& - \kappa \left( 1 - \lambda x \right).
    \end{aligned}
    \end{equation*}
    
    Thus, the exponential distribution can conform to 
    Eq. (\ref{eq:OptimalDiffusionWithKappa}) if $\forall t$ 
    \begin{equation*}
    \begin{aligned}
        \kappa = - \frac{1}{\lambda} \frac{d \lambda}{d t} 
        = - \frac{d \ln{\lambda}}{d t}.
    \end{aligned}
    \end{equation*}
    
\end{proof}

\subsection{Proof of Proposition \ref{prop:ExponentialDistributionEntropyIncrease}}

\begin{proof}
    For the exponential distribution, 
    \begin{equation*}
    \begin{aligned}
        H[f_E] =& 1 - \ln{\lambda}, \\
        \frac{d H}{d t} =& - \frac{d \ln{\lambda}}{d t}
    \end{aligned}
    \end{equation*}
    Hence, $\kappa = \kappa_0 = \text{const}$ implies, through Eq. (\ref{eq:KappaForExponentialDistribution}), 
    that $\frac{d H}{d t} = \kappa_0 = \text{const}$.
    It immediately follows from Eq. (\ref{eq:EntropyRateVsKullbackLeiblerRate}) that 
    $\sqrt{\frac{D_{\text{KL}}(f + \frac{\partial f}{\partial t} \dd{t} || f)}{\dd{t}^2}} = \kappa_0$.

\end{proof}

\end{document}